\documentclass[submission,copyright,creativecommons]{eptcs}
\providecommand{\event}{QPL 2024} 

\usepackage{iftex}
\usepackage{amsthm}
\usepackage{amsmath}
\usepackage{amssymb}
\usepackage{amsfonts}
\usepackage{mathtools}
\usepackage{mdframed}
\usepackage{graphicx}

\newcommand{\sem}{[-]}
\newcommand{\means}{\overset{\sem}{\mapsto}}
\usepackage{tikz}
\usepackage{tikz-cd}
\usetikzlibrary{graphs,decorations.pathmorphing,decorations.markings}
\usetikzlibrary{calc}
\usepackage{mathrsfs}
\usepackage{tikzit} 
\usepackage{caption}
\usepackage{subcaption}
\input{quantum.tikzdefs}

\tikzstyle{gate}=[shape=rectangle, text height=1.5ex, text depth=0.25ex, yshift=0.5mm, fill=white, draw=black, minimum height=3mm, yshift=-0.5mm, minimum width=3mm, font={\small}, tikzit category=circuit]
\tikzstyle{big gate}=[shape=rectangle, text height=1.5ex, text depth=0.25ex, yshift=0.5mm, fill=white, draw=black, minimum height=10mm, yshift=-0.5mm, minimum width=5mm, font={\small}, tikzit category=circuit]
\tikzstyle{Z dot}=[inner sep=0mm, minimum size=2mm, shape=circle, draw=black, fill=white, tikzit category=zx]
\tikzstyle{Z phase dot}=[minimum size=5mm, font={\footnotesize\boldmath}, shape=rectangle, rounded corners=2mm, inner sep=0.2mm, outer sep=-2mm, scale=0.8, tikzit shape=circle, draw=black, fill=white, tikzit draw=blue, tikzit category=zx]
\tikzstyle{X dot}=[Z dot, shape=circle, draw=black, fill=gray!40!white, tikzit category=zx]
\tikzstyle{X phase dot}=[Z phase dot, tikzit shape=circle, tikzit draw=blue, fill=gray!40!white, font={\footnotesize\boldmath}, tikzit category=zx]
\tikzstyle{hadamard}=[fill=yellow, draw=black, shape=rectangle, inner sep=0.6mm, minimum height=1.5mm, minimum width=1.5mm, tikzit category=zx]
\tikzstyle{paulibox}=[fill={rgb,255: red,221; green,221; blue,255}, draw=black, shape=rectangle, inner sep=0.6mm, minimum height=5mm, minimum width=5mm, font={\footnotesize}, text height=1.5ex, text depth=0.25ex, tikzit category=zx]
\tikzstyle{vertex}=[inner sep=0mm, minimum size=1mm, shape=circle, draw=black, fill=black, tikzit category=misc]
\tikzstyle{vertex set}=[inner sep=0mm, minimum size=1mm, shape=circle, draw=black, fill=white, font={\footnotesize\boldmath}, tikzit category=misc]
\tikzstyle{small black dot}=[fill=black, draw=black, shape=circle, inner sep=0pt, minimum width=1.2mm, tikzit category=circuit]
\tikzstyle{cnot ctrl}=[fill=black, draw=black, shape=circle, inner sep=0pt, minimum width=1.2mm, tikzit category=circuit]
\tikzstyle{cnot targ}=[fill=white, draw=white, shape=circle, tikzit category=circuit, label={center:$\oplus$}, inner sep=0pt, minimum width=2.1mm, tikzit fill={rgb,255: red,102; green,204; blue,255}, tikzit draw=black]
\tikzstyle{ket}=[fill=white, draw=black, shape=regular polygon, regular polygon sides=3, regular polygon rotate=-30, scale=0.7, inner sep=1pt, tikzit category=circuit, tikzit shape=rectangle, tikzit fill=green]
\tikzstyle{bra}=[fill=white, draw=black, shape=regular polygon, regular polygon sides=3, regular polygon rotate=30, scale=0.7, inner sep=1pt, tikzit category=circuit, tikzit shape=rectangle, tikzit fill=red]
\tikzstyle{xket}=[fill=gray!40!white, draw=black, shape=regular polygon, regular polygon sides=3, regular polygon rotate=60, scale=0.7, inner sep=1pt, tikzit category=circuit, tikzit shape=rectangle, tikzit fill=green]
\tikzstyle{xbra}=[fill=gray!40!white, draw=black, shape=regular polygon, regular polygon sides=3, regular polygon rotate=0, scale=0.7, inner sep=1pt, tikzit category=circuit, tikzit shape=rectangle, tikzit fill=red]
\tikzstyle{scalar}=[shape=rectangle, text height=1.5ex, text depth=0.25ex, yshift=0.5mm, fill=white, draw=black, minimum height=5mm, yshift=-0.5mm, minimum width=5mm, font={\small}]
\tikzstyle{clabel}=[fill=white, draw=none, shape=rectangle, tikzit fill={rgb,255: red,56; green,255; blue,242}, font={\footnotesize}, inner sep=1pt, tikzit category=labels]
\tikzstyle{empty diagram}=[draw={gray!40!white}, dashed, shape=rectangle, minimum width=1cm, minimum height=1cm, tikzit category=misc]
\tikzstyle{amap}=[fill=white, draw=black, shape=NEbox, tikzit category=asymmetric, tikzit fill=yellow, tikzit shape=rectangle]
\tikzstyle{amap conj}=[fill=white, draw=black, shape=NWbox, tikzit category=asymmetric, tikzit fill=green, tikzit shape=rectangle]
\tikzstyle{amap adj}=[fill=white, draw=black, shape=SEbox, tikzit category=asymmetric, tikzit fill=red, tikzit shape=rectangle]
\tikzstyle{amap trans}=[fill=white, draw=black, shape=SWbox, tikzit category=asymmetric, tikzit fill=orange, tikzit shape=rectangle]
\tikzstyle{astate}=[fill=white, draw=black, shape=NEtriangle, tikzit category=asymmetric, tikzit shape=circle, tikzit fill=yellow]
\tikzstyle{astate conj}=[fill=white, draw=black, shape=NWtriangle, tikzit category=asymmetric, tikzit shape=circle, tikzit fill=green]
\tikzstyle{astate adj}=[fill=white, draw=black, shape=SEtriangle, tikzit category=asymmetric, tikzit shape=circle, tikzit fill=red]
\tikzstyle{astate trans}=[fill=white, draw=black, shape=SWtriangle, tikzit category=asymmetric, tikzit shape=circle, tikzit fill=orange]
\tikzstyle{xastate}=[fill=gray!40!white, draw=black, shape=NEtriangle, tikzit category=asymmetric, tikzit shape=circle, tikzit fill=yellow]
\tikzstyle{xastate conj}=[fill=gray!40!white, draw=black, shape=NWtriangle, tikzit category=asymmetric, tikzit shape=circle, tikzit fill=green]
\tikzstyle{xastate adj}=[fill=gray!40!white, draw=black, shape=SEtriangle, tikzit category=asymmetric, tikzit shape=circle, tikzit fill=red]
\tikzstyle{xastate trans}=[fill=gray!40!white, draw=black, shape=SWtriangle, tikzit category=asymmetric, tikzit shape=circle, tikzit fill=orange]
\tikzstyle{blue label}=[text=blue]
\tikzstyle{lmat}=[shape=signal, signal to=west, signal from=east, fill={zx_grey}, draw=black, minimum height=6pt, inner sep=.75pt, font={\scriptsize \boldmath}, tikzit fill=gray, tikzit category=GLA]
\tikzstyle{rmat}=[lmat, shape=signal, signal to=east, signal from=west, tikzit fill=gray, tikzit category=GLA]
\tikzstyle{dmat}=[lmat, shape=signal, signal to=west, signal from=east, tikzit fill=gray, tikzit category=GLA, rotate=90]
\tikzstyle{umat}=[lmat, shape=signal, signal to=east, signal from=west, tikzit fill=gray, tikzit category=GLA, rotate=90]
\tikzstyle{d_split}=[shape=trapezium, fill=white, draw=black, inner sep=0pt, trapezium stretches body, text width=15pt, text height=7pt]
\tikzstyle{d_merge}=[{d_split}, shape=trapezium, draw=black, rotate=180]

\tikzstyle{hadamard edge}=[-, dashed, dash pattern=on 2pt off 0.5pt, thick, draw={rgb,255: red,68; green,136; blue,255}]
\tikzstyle{box edge}=[-, dashed, dash pattern=on 2pt off 0.5pt, thick, draw={rgb,255: red,203; green,192; blue,225}]
\tikzstyle{brace edge}=[-, tikzit draw=blue, decorate, decoration={brace,amplitude=1mm,raise=-1mm}]
\tikzstyle{diredge}=[->]
\tikzstyle{double edge}=[-, double, shorten <=-1mm, shorten >=-1mm, double distance=2pt]
\tikzstyle{gray edge}=[-, {gray!60!white}]
\tikzstyle{pointer edge}=[->, very thick, gray]
\tikzstyle{boldedge}=[-, line width=1.6pt, shorten <=-0.17mm, shorten >=-0.17mm]
\tikzstyle{bidir edge}=[<->, very thick, draw={rgb,255: red,191; green,191; blue,191}]
\tikzstyle{blue edge}=[-, blue]

\newtheorem*{claim}{Claim}

\newtheorem{theorem}{Theorem}[section]
\newtheorem{prop}{Proposition}[section]
\newtheorem{corollary}{Corollary}[section]
\newtheorem{lemma}{Lemma}[section]

\theoremstyle{definition}
\newtheorem{definition}{Definition}[section]

\newtheorem{algo}{Algorithm}[section]

\DeclareMathOperator\tr{tr}

\ifpdf
  \usepackage[english]{babel}        
  \usepackage[T1]{fontenc}        
\else
  \usepackage{breakurl}           
\fi

\title{The Qudit ZH Calculus for Arbitrary Finite Fields: \\ 
Universality and Application}
\author{Dichuan (David) Gao
\institute{Department of Computer Science\\
University of Oxford\\
Oxford, UK}
\email{dichuan.gao@maths.ox.ac.uk}
}
\def\titlerunning{Finite Field ZH}
\def\authorrunning{Dichuan (David) Gao}
\begin{document}
\maketitle

\begin{abstract}
We propose a generalization of the graphical ZH calculus to qudits of prime-power dimensions $q = p^t$, implementing field arithmetic in arbitrary finite fields. This is an extension of a previous result \cite{roy_qudit_2023} which implemented arithmetic of prime-sized fields; and an alternative to a result in \cite{de_beaudrap_simple_2023} which extended the ZH to implement cyclic ring arithmetic in $\mathbb Z / q\mathbb Z$ rather than field arithmetic in $\mathbb F_q$. We show this generalized ZH calculus to be universal over matrices $\mathbb C^{q^n} \to \mathbb C^{q^m}$ with entries in the ring $\mathbb Z[\omega]$ where $\omega$ is a $p$th root of unity. As an illustration of the necessity of such an extension of ZH for field rather than cyclic ring arithmetic, we offer a graphical description and proof for a quantum algorithm for polynomial interpolation. This algorithm relies on the invertibility of multiplication, and therefore can only be described in a graphical language that implements field, rather than ring, multiplication. 
\end{abstract}

\section{Introduction}
\label{section:introduction}

The ZH calculus is one of three popular graphical calculi for reasoning about quantum computation; the other two being the ZX and the ZW calculi \cite{carette_recipe_2020}. Each of these three calculi are useful for different types of reasoning: the ZX calculus represents the Clifford+Phase gateset most naturally \cite{coecke_picturing_2017, coecke_interacting_2011, cowtan_phase_2020, van_de_wetering_zx-calculus_2020}; the ZW calculus is able to represent sums of linear maps naturally with its W gate \cite{ wang_differentiating_2022, shaikh_how_2023, koch_quantum_2022} and helps us to reason about photonic and fermionic computations \cite{de_felice_diagrammatic_2019, hadzihasanovic_algebra_2017, wang_non-anyonic_2021}; and the ZH calculus yields the Toffoli+Hadamard gateset in a much simpler way than the other two \cite{backens_zh_2019, backens_completeness_2023, roy_qudit_2023}. 

At its core, the ZH calculus is able to represent the Toffoli gate elegantly because one of its primitive generators, the H-box, implements the multiplication operation "$\cdot$" on the computational basis (otherwise called the AND gate in the qubit case). This then makes it easy to obtain the Toffoli gate 
$$T = \sum_{x, y, z \in \mathbb F_2} |x, y, z \oplus (x \cdot y) \rangle \langle x, y, z |$$
in a diagrammatic form. 

For qubits, the ZH calculus has been shown to be universal \cite{backens_zh_2019}; and an elegant set of rewrite rules in the ZH calculus has been proven complete \cite{backens_completeness_2023}. The universality result has in turn been generalized to qudits of prime dimension $p$ by Roy et.al. \cite{roy_qudit_2023}. In qudits of prime dimension, owing to the existence of a ring isomorphism $\mathbb F_p \cong \mathbb Z / p\mathbb Z$, the above-mentioned AND gate naturally generalizes to the field multiplication in $\mathbb F_p$. The resulting graphical calculus remains universal, although no completeness result has been given. 

This qudit ZH calculus was further generalized by de Beaudrap and East to arbitrary dimensions $d$. \cite{de_beaudrap_simple_2023}. In their generalization, the generalized AND gate becomes ring multiplication in $\mathbb Z/d\mathbb Z$. But of course, for arbitrary finite fields $\mathbb F_q$ where $q$ is not prime, we have $\mathbb F_q \not \cong \mathbb Z / q \mathbb Z$ as rings. Thus, the generalized ZH calculus by de Beaudrap and East would not naturally implement field multiplication in these non-prime finite fields. This is why we present an alternative generalization of the qudit ZH calculus - one which implements the field multiplication of $\mathbb F_q$, rather than the ring multiplication of $\mathbb Z/q \mathbb Z$. The price of faithfully representing field multiplication of arbitrary finite fields, however, is that in non-prime dimensions we lose flexsymmetry of the H-box \cite{backens_zh_2019}. However, as we shall see in sections \ref{subsection:field-operations} and \ref{section:polynomial-interpolation}, this loss of flexsymmetry is only up to a dualizer, which corresponds to an algebraic operation, which turns out to be useful. 

In section \ref{section:the-calculus} we construct the finite field ZH calculus from its basic generators, and give some simple results characterizing this calculus. We then prove, in section \ref{section:universality}, the universality of this calculus over matrices of $\mathbb Z[\omega]$ where $\omega$ is a $p$th root of unity. Finally in section \ref{section:polynomial-interpolation} we demonstrate the usefulness of a finite field ZH calculus by diagrammatically proving the correctness of a quantum polynomial interpolation algorithm, which relies on the invertibility of multiplication in the computational basis - that is, the fact that it is a field. 

\section{The ZH Calculus for Arbitrary Finite Fields}
\label{section:the-calculus}

\subsection{Algebraic Background} 
\label{subsection:algebraic-background}

It is well known that if $p$ is a prime number, and $q = p^t$ for some $t \in \mathbb N$, then there exists a field $\mathbb F_q$ of cardinality $q$ and character $p$, namely the splitting field of the polynomial $g(x) = x^q - x$ over the prime field $\mathbb F_p$. 
Every finite field is of this form \cite{lang_algebra_2002}. Since splitting fields are unique up to isomorphism, we simply speak of $\mathbb F_q$ as \textit{the} finite field of size $q$. 

When $t = 1$, we have $\mathbb F_q = \mathbb F_p \cong \mathbb Z / p\mathbb Z$. But when $t \neq 1$, the field $\mathbb F_q$ is not cyclic, meaning $\mathbb F_q \not \cong \mathbb Z / q\mathbb Z$ as rings. There is then no group embedding $\mathbb F_q \hookrightarrow S^1$ where $S^1$ is the multiplicative group of the unit circle in the complex plane. This is a problem for the construction of a ZH calculus: one wants to define the H gate as the map 
$$\input{./figures/H-spider.tikz} \overset{?}{=} \frac{1}{\sqrt{q}}\sum_{\substack{j_1, ..., j_m \in \mathbb F_q \\ i_1, ..., i_n \in \mathbb F_q}} \omega^{j_1 ... j_m i_1 ... i_n} | j_1 \rangle ... | j_m \rangle \langle i_1 | ... \langle i_n |$$
where $\omega$ is a root of unity. But this is nonsensical, for what does it mean to raise $\omega$ to the power of an element of a field $\mathbb F_q$ which has no interpretation as $\mathbb Z / q\mathbb Z$? To make sense of that, there needs to be a map $j \mapsto \omega^j$ which is a group embedding $\mathbb F_q \hookrightarrow S^1$; but there is no such embedding. This is essentially the complaint in Roy et.al. that prevented a full generalization of ZH to arbitrary dimensions \cite{roy_qudit_2023}. 

This difficulty can be overcome, however, if we relax the requirement for a group \textit{embedding}, and consider how the standard definition of the Galois Qudit Pauli group \cite{gottesman_surviving_2024} reduces the information of $\mathbb F_q$ to $\mathbb F_p$ in its definition of the Pauli Z. We briefly recall that construction here: 

\begin{definition}
    For a finite field $\mathbb F_q$ with $q = p^t$, the field trace is a $\mathbb F_p$ linear map $\tr: \mathbb F_q \to \mathbb F_p$, defined by 
    \begin{equation}
        \tr x := x + x^p + \dots + x^{p^{(t-1)}}
    \end{equation}
\end{definition}

Then the Pauli Group of a Galois Qudit is defined by the following gates: 

\begin{definition} (\cite{gottesman_surviving_2024})
    For $\beta \in \mathbb F_q$, define $X^\beta, Z^\beta: \mathbb C^q \to \mathbb C^q$ as acting on the computational basis, labelled with elements of $\mathbb F_q$, as: 
    \begin{equation}
        \begin{split}
            X^\beta|\gamma\rangle &:= |\gamma + \beta\rangle \\ 
            Z^\beta|\gamma\rangle &:= \omega^{\tr(\beta\gamma)}|\gamma\rangle 
        \end{split}
    \end{equation}
    where $\gamma \in \mathbb F_q$, and $\omega = \exp(2\pi i / p)$ is a $p$th root of unity. 
\end{definition}

This suggests that the H-box for Galois Qudits can also be defined in terms of the field trace. Indeed the following essential feature of the H-box holds: 
\begin{lemma}
    \label{lemma:sum-on-circle}
    If $j \in \mathbb F_q$; and $\omega$ is a primitive $p$th root of unity in $\mathbb C$, then 
    $$\sum_{k \in \mathbb F_q} \omega^{\tr(jk)} = \begin{cases}
        0 &j \neq 0 \\ 
        q &j = 0
    \end{cases}$$
\end{lemma}

\begin{proof}
    First suppose $j \neq 0$. Then the map $k \mapsto jk$ is bijective on $\mathbb F_q$. Therefore: 
    $$\sum_{k \in \mathbb F_q} \omega^{\tr (jk)} = \sum_{k' \in \mathbb F_q} \omega^{\tr(k')} = p^{t-1} \cdot \sum_{\theta = 0}^{p-1} \omega^\theta = 0$$
    where the second equality follows from the fact that the field trace is surjective and linear over $\mathbb F_p$. 
    
    The other case where $j = 0$ is straight forward as $\tr(0) = 0$. 
\end{proof}

Before we move onto defining our graphical calculus, we remind the reader of the following algebraic theorem:

\begin{theorem}
    \label{thm:existence_of_xi} 
    (Cohen and Hachenberger 1999 \cite{cohen_primitive_1999}). For any finite field $\mathbb F_q$ with $q = p^t$, there exists an element $\xi \in \mathbb F_q$ such that: 
    \begin{enumerate}
        \item $\xi$ is primitive; i.e. $\{1, \xi, \dots, \xi^{q-2}\} = \mathbb F_q^\times$; 
        \item $\xi$ is normal; i.e. $\{\xi^{p^0}, \dots, \xi^{p^{t-1}}\}$ forms a basis of $\mathbb F_q$ over $\mathbb F_p$; 
        \item $\tr \xi = 1$.  
    \end{enumerate}
\end{theorem}

\subsection{Generators of the Finite Field ZH} 
\label{subsection:generators}

\begin{definition}
    For a field $\mathbb F_q$ with $q = p^t$, and for a particular choice of a primitive, normal, trace-one element $\xi \in \mathbb F_q$ (which exists by virtue of \ref{thm:existence_of_xi}), the Galois-Qudit ZH calculus over $\mathbb F_q$ is the dagger-PROP generated by three families of generators, with the following intended semantic maps: 
    \begin{equation*}
        \input{./figures/Z-spider.tikz} \mapsto \sum_{i \in \mathbb F_q} |i \rangle^{\otimes m} \langle i |^{\otimes n} 
    \end{equation*} 
    \begin{equation*}
        \input{./figures/H-spider.tikz} \mapsto \frac{1}{\sqrt{q}}\sum_{\substack{j_1, ..., j_m \in \mathbb F_q \\ i_1, ..., i_n \in \mathbb F_q}} \omega^{\tr(j_1...j_m i_1...i_n)} | j_1 \rangle ... | j_m \rangle \langle i_1 | ... \langle i_n |
    \end{equation*}
    \begin{equation*}
        \begin{tikzpicture}
	\begin{pgfonlayer}{nodelayer}
		\node [style=X phase dot] (0) at (-0.75, 0) {$\xi$};
		\node [style=none] (1) at (0.75, 0) {};
	\end{pgfonlayer}
	\begin{pgfonlayer}{edgelayer}
		\draw (0) to (1.center);
	\end{pgfonlayer}
\end{tikzpicture}
 \mapsto \sqrt{q}|\xi\rangle. 
    \end{equation*}
    where $\omega = e^{2\pi i/p}$. The first family is called the Z-spider, the second family is called the H-spider or H-box, and the third is called the $\xi$-state or $\xi$-lollipop. 
\end{definition}

Notice that only the specific state $|\xi\rangle$ is given the privileged status of being a generator of our calculus. A reader familiar with the ZH calculus for qubits may wonder why this is so. One answer is that this is the "stamp" in which the none-uniqueness of the vector-space-structure of $\mathbb F_q$ over $\mathbb F_p$ is encoded. A more pragmatic answer is made clear in section \ref{subsection:field-operations}, and particularly lemma \ref{lemma:basis_states_are_in_ZH}: in the case of non-cyclic fields, this extra generator is needed for universality. 

The Z-spider here is just the usual Z-spider in the ZX-calculus, so all the usual properties from ZX-calculus hold here. The H-spider is a kind of generalized discrete Fourier transform. Thus we immediately obtain a derived generator corresponding to the inverse Fourier transform: 

\begin{prop}
    The family of H-spiders is self-transpose (up to switching $m$ and $n$); therefore its conjugate is its adjoint, and is given by "inverse Fourier transform": 
    \begin{equation*}
        \input{./figures/H-dagger.tikz} \mapsto \frac{1}{\sqrt{q}}\sum_{\substack{j_1, ..., j_m \in \mathbb F_q \\ i_1, ..., i_n \in \mathbb F_q}} \omega^{-\tr( j_1 ... j_m i_1 ... i_n )} | j_1 \rangle ... | j_m \rangle \langle i_1 | ... \langle i_n |
    \end{equation*}
\end{prop}

It follows also that the one-in-one-out H-spider is unitary, since its adjoint is manifestly its inverse. 

\begin{proof}
    The fact that the family of H-spider is self-transpose is just another way to say that multiplication is commutative in a field. The above equation describing the adjoint and conjugate of the H-spider amounts to taking the complex conjugate of each of its matrix entries. 
\end{proof}

These are the primitive generators of the phaseless Galois-Qudit ZH calculus. The \textit{phased} (or \textit{labelled}) Galois-Qudit ZH calculus allows further labelling of the H-box by elements of a ring: 

\begin{definition}
    For a field $\mathbb F_q$ with $q = p^t$, and $\omega = e^{2\pi i/p}$, let $R$ be a ring with $\mathbb Z[\omega] \subset R$. The $R$-\textbf{phased Galois-Qudit ZH-calculus} on $\mathbb F_q$ is a family of linear maps generated by the phaseless Galois-Qudit ZH-calculus on $\mathbb F_q$, together with labelled H-spiders of the following form: for every $r \in R$, 
    $$\input{./figures/Hr-spider.tikz} \mapsto \frac{1}{\sqrt{q}}\sum_{\substack{j_1, ..., j_m \in \mathbb F_q \\ i_1, ..., i_n \in \mathbb F_q}} r^{\tr ( j_1 ... j_m  i_1 ... i_n )} | j_1 \rangle ... | j_m \rangle \langle i_1 | ... \langle i_n |$$
\end{definition}

In both the phaseless and the phased calculi we have some useful gadgets derived from the generators described above:
    \begin{equation*}
        \input{./figures/X-spider.tikz} = \input{./figures/X-spider-constr.tikz} \hspace{15mm}
        \begin{tikzpicture}
	\begin{pgfonlayer}{nodelayer}
		\node [style=scalar] (0) at (0, 0) {$q^{\frac{1}{2}}$};
	\end{pgfonlayer}
\end{tikzpicture}
 = \begin{tikzpicture}
	\begin{pgfonlayer}{nodelayer}
		\node [style=X dot] (0) at (-0.5, 0) {};
		\node [style=Z dot] (1) at (0.5, 0) {};
	\end{pgfonlayer}
	\begin{pgfonlayer}{edgelayer}
		\draw (0) to (1);
	\end{pgfonlayer}
\end{tikzpicture}
 \hspace{15mm} \begin{tikzpicture}
	\begin{pgfonlayer}{nodelayer}
		\node [style=scalar] (0) at (0.5, 0) {$q^{-\frac{1}{2}}$};
	\end{pgfonlayer}
\end{tikzpicture}
 = \begin{tikzpicture}
	\begin{pgfonlayer}{nodelayer}
		\node [style=X dot] (0) at (-1, 0) {};
		\node [style=Z dot] (1) at (0, 0) {};
		\node [style=hadamard] (2) at (1, 1) {};
		\node [style=hadamard] (3) at (1, -1) {};
	\end{pgfonlayer}
	\begin{pgfonlayer}{edgelayer}
		\draw (0) to (1);
		\draw [bend left=15] (1) to (2);
		\draw [bend right=15] (1) to (3);
	\end{pgfonlayer}
\end{tikzpicture}

    \end{equation*}
The left-most one is called the X-spider, while the other two are the scalars. By taking repeated tensor-products of the two scalars above, we can construct any scalars of form $q^{\pm \frac{n}{2}}$. 

\subsection{Finite Field Arithmetics from ZH} 
\label{subsection:field-operations}

As we mentioned in the introduction (\ref{section:introduction}), the power of the ZH calculus comes in part from the fact that it effortlessly implements both addition (XOR) and multiplication (AND) on the computational basis. We now describe the way in which this holds in our ZH calculus for non-prime sized finite fields. All propositions in this section are proven in appendix \ref{append:arithmetics}. 

\begin{prop}
    \label{prop:neg} 
    The 1-to-1 X-spider is the same as applying the H-box twice, and acts on the computational basis as negation: 
    \begin{equation}
        \label{eq:neg}
        \begin{tikzpicture}
	\begin{pgfonlayer}{nodelayer}
		\node [style=X dot] (0) at (0, 0) {};
		\node [style=none] (1) at (-1, 0) {};
		\node [style=none] (2) at (1, 0) {};
	\end{pgfonlayer}
	\begin{pgfonlayer}{edgelayer}
		\draw (0) to (1.center);
		\draw (0) to (2.center);
	\end{pgfonlayer}
\end{tikzpicture}
 = \begin{tikzpicture}
	\begin{pgfonlayer}{nodelayer}
		\node [style=hadamard] (0) at (-0.5, 0) {};
		\node [style=hadamard] (1) at (0.5, 0) {};
		\node [style=none] (3) at (1.5, 0) {};
		\node [style=none] (4) at (-1.5, 0) {};
	\end{pgfonlayer}
	\begin{pgfonlayer}{edgelayer}
		\draw (0) to (1);
		\draw (1) to (3.center);
		\draw (4.center) to (0);
	\end{pgfonlayer}
\end{tikzpicture}
 = \sum_{i \in \mathbb F_q} |-i\rangle \langle i|  \tag{neg}
    \end{equation}
\end{prop}

\begin{prop}
    \label{prop:add} 
    The 2-to-1 X-spider followed by a 1-to-1 X-spider implements addition in $\mathbb F_q$, and the 0-to-1 X-spider is the additive identity $0 \in \mathbb F_q$ 
    \begin{equation}
        \label{eq:add}
        \input{./figures/Addition.tikz} = \sum_{i, j \in \mathbb F_q}|i+j\rangle \langle i | \langle j |  \hspace{15mm} \begin{tikzpicture}
	\begin{pgfonlayer}{nodelayer}
		\node [style=X dot] (0) at (0, 0) {};
		\node [style=none] (1) at (1, 0) {};\
        \node [style=scalar] (2) at (-1.5, 0) {$q^{-\frac{1}{2}}$}; 
	\end{pgfonlayer}
	\begin{pgfonlayer}{edgelayer}
		\draw (0) to (1.center);
	\end{pgfonlayer}
\end{tikzpicture}
 = |0\rangle \tag{+}
    \end{equation}
\end{prop}

\begin{prop}
    \label{prop:mult}
    The 2-to-1 H-spider followed by an $H^\dag$ implements multiplication in $\mathbb F_q$, and the 0-to-1 H-spider followed by an $H^\dag$ is the multiplicative identity $1 \in \mathbb F_q$ 
    \begin{equation}
        \label{eq:mult}
        \input{./figures/Multiplication.tikz} = \sum_{i, j \in \mathbb F_q}|i\cdot j\rangle \langle i | \langle j | \hspace{15mm} \begin{tikzpicture}
	\begin{pgfonlayer}{nodelayer}
		\node [style=hadamard] (0) at (-0.5, 0) {};
		\node [style=hadamard] (1) at (0.5, 0) {$\dag$};
		\node [style=none] (2) at (1.5, 0) {};
	\end{pgfonlayer}
	\begin{pgfonlayer}{edgelayer}
		\draw (0) to (1);
		\draw (1) to (2.center);
	\end{pgfonlayer}
\end{tikzpicture}
 = |1\rangle \tag{m}
    \end{equation}
\end{prop}

With the arithmetic operations in hand, we can therefore represent the natural action of the Galois group $\Gamma_{q/p}$ on the qudit of dimension $q$ by: 
\begin{equation}
    \Gamma_{q/p} \cong \left\{\input{./figures/galois-group-rep.tikz} \middle \vert \tau = 0, \dots, t-1\right\}
\end{equation}
and so we can implement the following field-trace map: 
\begin{equation}
    \input{./figures/field-trace.tikz} \quad := \quad \input{./figures/field-trace-constr.tikz} \quad \means \quad \sum_{j \in \mathbb F_q} |\tr j\rangle \langle j |. 
\end{equation}
Note that this field-trace map is a classical map (it maps the computational basis to the computational basis), and therefore is self-conjugate. But it is not injective, and therefore we introduce the diagrammatic asymmetry to denote its directionality.




A major difference between the Galois-Qudit ZH and the prime-dimensional ZH is that the additive structure on the computational basis is non-cyclic; and therefore one \textit{cannot} obtain all computational basis states by simply starting from $|0\rangle$ and applying the Pauli X gate to it. 
Instead we construct them as follows: 

\begin{lemma}
    \label{lemma:basis_states_are_in_ZH}
    For any Galois field $\mathbb F_q$ and any $j \in \mathbb F_q$, there exists a diagram in the phase-free Galois-Qudit ZH calculus that represents $|j\rangle$. 
\end{lemma}

\begin{proof}
    We have already remarked in proposition \ref{prop:add} that the scaled X-lollipop $\begin{tikzpicture}
	\begin{pgfonlayer}{nodelayer}
		\node [style=X dot] (0) at (0, 0) {};
		\node [style=none] (1) at (1, 0) {};\
        \node [style=scalar] (2) at (-1.5, 0) {$q^{-\frac{1}{2}}$}; 
	\end{pgfonlayer}
	\begin{pgfonlayer}{edgelayer}
		\draw (0) to (1.center);
	\end{pgfonlayer}
\end{tikzpicture}
$ is the basis state $|0\rangle$. So it suffices to construct all the basis states corresponding to nonzero elements of $\mathbb F_q$. 

    Since $\xi$ is a primitive element of $\mathbb F_q$, so the powers of $\xi$ range over all the nonzero elements: $\{1, \xi, \dots, \xi^{q-2}\} = \mathbb F_q^\times$. Thus the following states range over all the nonzero basis states: 
    $$\left\{\input{./figures/nth-power-of-xi.tikz} \, \mid \, n = 0, \dots, q-2 \right\} \mapsto \{|j\rangle \, \mid \, j \in \mathbb F_q^\times\}.$$
\end{proof}


Since, then, every computational basis state is constructible in the phase-free finite field ZH calculus, we will simply denote 
\begin{equation}
\label{eq:basis-lollipop-def}
    \begin{tikzpicture}
	\begin{pgfonlayer}{nodelayer}
		\node [style=X phase dot] (0) at (-0.5, 0) {$j$};
		\node [style=none] (1) at (0.5, 0) {};
	\end{pgfonlayer}
	\begin{pgfonlayer}{edgelayer}
		\draw (0) to (1.center);
	\end{pgfonlayer}
\end{tikzpicture}
 := q^{\frac{1}{2}}|j\rangle \hspace{10mm} \begin{tikzpicture}
	\begin{pgfonlayer}{nodelayer}
		\node [style=Z phase dot] (0) at (0, 0) {$j$};
		\node [style=none] (1) at (1, 0) {};
	\end{pgfonlayer}
	\begin{pgfonlayer}{edgelayer}
		\draw (0) to (1.center);
	\end{pgfonlayer}
\end{tikzpicture}
 := \begin{tikzpicture}
	\begin{pgfonlayer}{nodelayer}
		\node [style=X phase dot] (0) at (-1, 0) {$j$};
		\node [style=hadamard] (1) at (0, 0) {$\dag$};
		\node [style=none] (2) at (1, 0) {};
	\end{pgfonlayer}
	\begin{pgfonlayer}{edgelayer}
		\draw (0) to (1);
		\draw (1) to (2.center);
	\end{pgfonlayer}
\end{tikzpicture}
 \tag{bas}
\end{equation}
for any $j \in \mathbb F_q$. We emphasize that these decorated lollipops are in the phase-\textit{free} calculus. The Z-lollipop states are \textit{not} self-conjugate; and so their transposes are not their adjoints. Indeed the conjugate of a Z-lollipop state is 
$$\overline{\left(\begin{tikzpicture}
	\begin{pgfonlayer}{nodelayer}
		\node [style=Z phase dot] (0) at (0, 0) {$j$};
		\node [style=none] (1) at (1, 0) {};
	\end{pgfonlayer}
	\begin{pgfonlayer}{edgelayer}
		\draw (0) to (1.center);
	\end{pgfonlayer}
\end{tikzpicture}
\right)} = \begin{tikzpicture}
	\begin{pgfonlayer}{nodelayer}
		\node [style=X phase dot] (0) at (-1, 0) {$j$};
		\node [style=hadamard] (1) at (0, 0) {};
		\node [style=none] (2) at (1, 0) {};
	\end{pgfonlayer}
	\begin{pgfonlayer}{edgelayer}
		\draw (0) to (1);
		\draw (1) to (2.center);
	\end{pgfonlayer}
\end{tikzpicture}
 = \input{./figures/white-basis-state-conj-constr2.tikz} = \begin{tikzpicture}
	\begin{pgfonlayer}{nodelayer}
		\node [style=Z phase dot] (0) at (0, 0) {$-j$};
		\node [style=none] (1) at (1, 0) {};
	\end{pgfonlayer}
	\begin{pgfonlayer}{edgelayer}
		\draw (0) to (1.center);
	\end{pgfonlayer}
\end{tikzpicture}
$$
The following corollary follows directly from the definitions of the X- and Z-lollipops: 
\begin{corollary}
    \label{cor:lollipops_as_fourier_bases}
    The X- and Z-lollipop states, up to a factor of $\sqrt{q}$, are orthonormal bases of $\mathbb C^q$, and are Fourier with respect to each other. That is: 
    $$\begin{tikzpicture}
	\begin{pgfonlayer}{nodelayer}
		\node [style=Z phase dot] (0) at (-1, 0) {$i$};
		\node [style=Z phase dot] (1) at (1, 0) {$-j$};
	\end{pgfonlayer}
	\begin{pgfonlayer}{edgelayer}
		\draw (0) to (1);
	\end{pgfonlayer}
\end{tikzpicture}
 = \begin{tikzpicture}
	\begin{pgfonlayer}{nodelayer}
		\node [style=X phase dot] (0) at (-1, 0) {$i$};
		\node [style=X phase dot] (1) at (1, 0) {$j$};
	\end{pgfonlayer}
	\begin{pgfonlayer}{edgelayer}
		\draw (0) to (1);
	\end{pgfonlayer}
\end{tikzpicture}
 = q\delta_i^j \hspace{10mm} \begin{tikzpicture}
	\begin{pgfonlayer}{nodelayer}
		\node [style=Z phase dot] (0) at (-1, 0) {$i$};
		\node [style=X phase dot] (1) at (1, 0) {$j$};
	\end{pgfonlayer}
	\begin{pgfonlayer}{edgelayer}
		\draw (0) to (1);
	\end{pgfonlayer}
\end{tikzpicture}
 = q^{\frac{1}{2}}\omega^{-\tr(ij)}. $$
\end{corollary}


\subsection{Rewrite Rules}
\label{subsection:rewrite-rules}

In Backens et.al. \cite{backens_completeness_2023} a set of eight rewrite rules, each intuitively corresponding to an algebraic fact about $\mathbb Z/2\mathbb Z$, were introduced for the qubit ZH calculus. In figure \ref{fig:rewrite-rules}, similarly, we present a set of nine rules for the finite field ZH calculus, each corresponding to an algebraic fact about $\mathbb F_q$. 

\begin{figure}[h]
    \centering
    \begin{subfigure}[b]{0.35\textwidth}
        \centering
        \input{./figures/rule-zs-lhs.tikz} = \input{./figures/rule-zs-rhs.tikz} 
        \caption*{(zs)}
    \end{subfigure}
    \begin{subfigure}[b]{0.3\textwidth}
        \centering
        \input{./figures/rule-id-lhs.tikz} = \input{./figures/Wire.tikz} 
        \caption*{(id)}
    \end{subfigure}
    \begin{subfigure}[b]{0.3\textwidth}
        \centering        \input{./figures/rule-ch-lhs.tikz} = \input{./figures/rule-ch-rhs.tikz} 
        \caption*{(ch)}
    \end{subfigure}
    \begin{subfigure}[b]{0.35\textwidth}
        \centering
        \input{./figures/rule-hs-lhs.tikz} = \input{./figures/rule-hs-rhs.tikz} 
        \caption*{(hs)}
    \end{subfigure}
    \begin{subfigure}[b]{0.3\textwidth}
        \centering 
        \input{./figures/rule-spl-lhs.tikz} = \input{./figures/Wire.tikz} 
        \caption*{(spl)}
    \end{subfigure}
        \begin{subfigure}[b]{0.3\textwidth}
        \centering
        \input{./figures/copy-thru-white-lhs.tikz} = \input{./figures/copy-thru-white-rhs.tikz} 
        \caption*{(cp)}
    \end{subfigure}
    \begin{subfigure}[b]{0.35\textwidth}
        \centering
        \input{./figures/rule-ba1-lhs.tikz} = \input{./figures/rule-ba1-rhs.tikz} 
        \caption*{(ba1)}
    \end{subfigure}
    \begin{subfigure}[b]{0.35\textwidth}
        \centering
        \input{./figures/rule-ba2-lhs.tikz} = \input{./figures/rule-ba2-rhs.tikz} 
        \caption*{(ba2)}
    \end{subfigure}
    \begin{subfigure}[b]{0.25\textwidth}
        \centering
        \input{./figures/rule-pm-lhs.tikz} = \input{./figures/rule-pm-rhs.tikz}
        \caption*{(pm)}
    \end{subfigure}
    \caption{Rewrite Rules of the Finite Field ZH}
    \label{fig:rewrite-rules}
\end{figure}

The soundness of these rules is proven in appendix \ref{append:rewrite}. Note that we have not mentioned here the rule that four fourier transforms in a row makes the identity. This is because this rule is derivable from the other rules shown in figure \ref{fig:rewrite-rules}. The derivation is identical to that found in section 3 of \cite{roy_qudit_2023}. 

Another convenient derived rule to note is the fourier-copy rule: for $j\in \mathbb F_q$, 
\begin{equation}
\label{eq:fourier-copy-rule}
    \input{./figures/copy-thru-grey-lhs.tikz} = \input{./figures/copy-thru-grey-rhs.tikz} \tag{cpx}
\end{equation}
which follows straightforwardly from the copy rule (cp) and the definition of the white lollipop (\ref{eq:basis-lollipop-def}). 

\section{Universality of the Phased Finite Field ZH} 
\label{section:universality}

The generators of the phase-free finite field ZH all have matrices of the form $q^{\frac{u}{2}} M$, where $u \in \mathbb Z$, and $M$ is a matrix with entries in the ring $R = \mathbb Z\left[\omega\right]$, $\omega$ being a $p$th root of unity. The same is true for the phased finite field ZH calculus whose phase ring is $R$. By universality, then, we mean the ability to express any linear transformation $\mathbb C^{q^n} \to \mathbb C^{q^m}$ whose matrix (with respect to the computational basis $\mathbb F_q$) has entries in $R$. 

\begin{theorem}
    \label{thm:universality_phased}
    Let $q = p^t$ for some prime $p$ and $t \in \mathbb N$, and let $\omega = e^{\frac{2\pi i}{p}}$. For any linear map $M: \mathbb C^{q^n} \to \mathbb C^{q^m}$
    $$M = \sum_{\substack{j_1, ..., j_m \in \mathbb F_q \\ i_1, ..., i_n \in \mathbb F_q}} m_{i_1, ..., i_n}^{j_1, ..., j_m}| j_1 \rangle ... | j_m \rangle \langle i_1 | ... \langle i_n |$$
    where each entry $m_{i_1, ..., i_n}^{j_1, ..., j_m} \in R$, there exists a diagram in the phased finite field ZH calculus for $\mathbb F_q$ that represents $M$. 
\end{theorem}

\begin{proof}
    In the special case where $q = p$, this is proven in \cite{roy_qudit_2023}. The proof is constructive and, in our case, goes as follows: 

\begin{enumerate}
    \item Any matrix in $R^{m \times n}$ is equal to the Schur product of pseudo-binary matrices; that is, matrices whose entries are either $1$ or $r$ for some $r \in R$. Therefore the problem of expressing any matrix over $R$ is reducible to the problem of 
    \begin{enumerate}
        \item Expressing any pseudo-binary matrix over $R$; and 
        \item Expressing the Schur product of matrices. 
    \end{enumerate}
    \item The Z-spider suffices to implement Schur products. So the problem is reduced to that of expressing any pseudo-binary matrix. 
    \item Any pseudo-binary matrix $M: \mathbb C^{q^n} \to \mathbb C^{q^m}$ can be encoded as a propositional formula $\phi$ with $n+m$ variables over $(\mathbb F_q, +, -, \cdot, =)$ such that $\phi(i_1, ..., i_n, j_1, ..., j_m) \leftrightarrow m_{i_1, ..., i_n}^{j_1, ..., j_m} = 1$.   
    \item Any propositional formula $\phi(x_1, ..., x_n)$ corresponds to a polynomial $f_\phi \in \mathbb F_q[x_1, ..., x_n]$ such that $f_\phi(i_1, ..., i_n) = 0 \leftrightarrow \phi(i_1, ..., i_n)$. 
    \item For any polynomial $f \in \mathbb F_q[x_1, ..., x_n]$, there exists an $n$-input $1$-output phase-free ZH diagram $D_f$ that acts on the computational basis as:
    $$D_f|i_1...i_n\rangle = \begin{cases}
        |0\rangle &f(i_1, ..., i_n) = 0 \\ 
        |1\rangle &f(i_1, ..., i_n) \neq 0
    \end{cases} $$
    \item For any $r \in R$, then, we can post-select on $D_f$ with the phased H-box $\begin{tikzpicture}
	\begin{pgfonlayer}{nodelayer}
		\node [style=hadamard] (0) at (0.5, 0) {$r$};
		\node [style=none] (1) at (-0.5, 0) {};
	\end{pgfonlayer}
	\begin{pgfonlayer}{edgelayer}
		\draw (1.center) to (0);
	\end{pgfonlayer}
\end{tikzpicture}
 = \sum_{j} r^{\tr(j)} \langle j |$ to obtain a (phased) ZH diagram $D$ with $n$ inputs and no outputs, such that 
    $$D|i_1 ... i_n\rangle = \begin{cases}
        1 &f(i_1, ..., i_n) = 0 \\ 
        r &f(i_1, ..., i_n) \neq 0
    \end{cases}$$
    \item Using facts 3-6 together, given any pseudo-binary matrix $M$ that is $q^n \times q^m$, there is a (phased) ZH diagram $D$ with $n+m$ inputs and no outputs such that $D|i_1, ..., i_n, j_1, ..., j_m\rangle = m_{i_1, ..., i_n}^{j_1, ..., j_m}$. The corresponding $n \to m$ diagram given by the Choi-Jamiolkowski isomorphism (bending the last $m$ wires) then implements $M$. 
\end{enumerate}

The proofs for facts 1, 2, 3, 6, and 7 are identical to that in the case of prime-dimensional qudits. Since for these facts no change is required whatsoever, we will not set out the proofs again here, and refer readers to section 4 of \cite{roy_qudit_2023}. What does change, however, are the proofs for the facts 4 and 5. We prove these below as lemmas. Taken together, this completes the proof for theorem \ref{thm:universality_phased}.
\end{proof}

\begin{lemma}
    Let $\phi(x_1, ..., x_n)$ be a propositional formula on $(\mathbb F_q, +, -, \cdot, =)$. There exists a polynomial $f_\phi \in \mathbb F_q[x_1, ..., x_n]$ such that for all $i_1, ..., i_n \in \mathbb F_q$, $f_\phi(i_1, ..., i_n) = 0 \leftrightarrow \phi(i_1, ..., i_n)$. 
\end{lemma}

\begin{proof}
    We prove by induction on the complexity of $\phi$. Any term with $n$ variables in $(\mathbb F_q, +, -, \cdot, =)$ is by construction a polynomial in $\mathbb F_q[x_1, ..., x_n]$. The only relation symbol is $=$. So if $\phi$ is atomic, it must be of form $\phi = (f_1(x_1, ..., x_n) = f_2(x_1, ..., x_n))$ for polynomials $f_1, f_2 \in \mathbb F_q[x_1, ..., x_n]$. Then $f_\phi = f_1 - f_2$ satisfies the requirement of the theorem. Now for the inductive step: 
    \begin{enumerate}
        \item If $\phi = \neg \phi'$ for some formula $\phi'$, then set $f_\phi = 1 - (f_{\phi'})^{q-1}$. Then, if $\phi'$ is true, we have $f_\phi = 1 - 0 = 1$, so $f_\phi \neq 0$ when $\phi$ is false. Conversely, if $\phi'$ is false, we have $f_{\phi'} \neq 0$, and so $(f_{\phi'})^{q-1} = 1$, and so $f_\phi = 0$. 
        \item If $\phi = \phi' \vee \phi''$ for formulae $\phi', \phi''$, then set $f_\phi = f_{\phi'} \cdot f_{\phi''}$. Then we have $f_\phi = 0 \leftrightarrow ((f_{\phi'} = 0) \vee (f_{\phi''} = 0)) \leftrightarrow (\phi' \vee \phi'') = \phi$. 
    \end{enumerate}
    This completes the induction. 
\end{proof}

\begin{lemma}
    For any polynomial $f \in \mathbb F_q[x_1, ..., x_n]$ there exists a phase-free ZH diagram $D_f$ that acts on the computational basis as 
    $$D_f|i_1...i_n\rangle = \begin{cases}
        |0\rangle &f(i_1, ..., i_n) = 0 \\ 
        |1\rangle &f(i_1, ..., i_n) \neq 0
    \end{cases}$$
\end{lemma}

\begin{proof}
    The construction proceeds in two steps. First, we construct a diagram $D'_f$ with $n$ input wires and $1$ output wire such that $D'_f|i_1, ..., i_n \rangle = |f(i_1, ..., i_n) \rangle$. We then post-compose $D'_f$ with the map 
    $$\input{./figures/q-1-power.tikz} = \sum_{j \in \mathbb F_q} |j^{q-1}\rangle \langle j |$$
    which sends $|0\rangle \mapsto |0\rangle$, and all nonzero $|j\rangle \mapsto |1\rangle$. This gives us $D_f$. 

    So it remains only to construct $D'_f$, the diagram that acts on the computational basis as the polynomial itself. We do this by induction on $n$. If $n = 0$ then $f$ is just a constant $f = j \in \mathbb F_q$. Lemma \ref{lemma:basis_states_are_in_ZH} gives us a phase-free ZH diagram $\begin{tikzpicture}
	\begin{pgfonlayer}{nodelayer}
		\node [style=X phase dot] (0) at (-0.5, 0) {$j$};
		\node [style=none] (1) at (0.5, 0) {};
	\end{pgfonlayer}
	\begin{pgfonlayer}{edgelayer}
		\draw (0) to (1.center);
	\end{pgfonlayer}
\end{tikzpicture}
 = q^{-\frac{1}{2}}|j\rangle$. 
    So constants are constructible in the phase-free ZH calculus. 

    Now for the inductive step: since $\mathbb F_q[x_1, ..., x_n] = \mathbb F_q[x_1, ..., x_{n-1}][x_n]$, so $f$ can be written as $f = f_0 + f_1 x_n + ... + f_k x_n^k$ for some collection $f_0, ..., f_k$ of polynomials in $\mathbb F_q[x_1, ..., x_{n-1}]$. By the inductive hypothesis there are phase-free ZH diagrams corresponding to each $f_i$. Then the following phase-free diagram constructs $f$: 
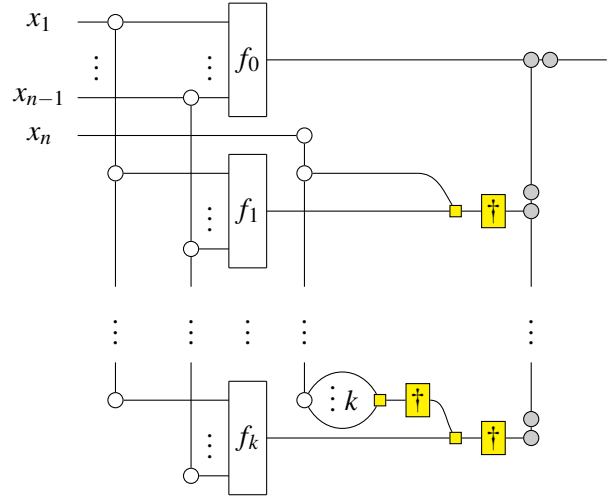
    $$\input{./figures/big-polynomial-constr.tikz}$$
    The correctness of this construction is readily verified by propositions \ref{prop:add} and \ref{prop:mult}. 
\end{proof}

\section{An Example: The Polynomial Interpolation Algorithm}
\label{section:polynomial-interpolation}
Suppose we are given oracular access to a polynomial $f \in \mathbb F_q[x]$ of known degree $d$: $f(x) = a_d x^d + ... + a_1 x + a_0$, but we do not know the coefficients $a_d, ..., a_0$. The polynomial interpolation problem asks us to infer the values of the coefficients with the fewest possible queries to the oracle. 

On classical computers it is known that $d+1$ queries are needed, and is also clearly sufficient: for that would provide us with $d+1$ linear equations between $d+1$ unknown quantities. For quantum computers, it was initially shown by Kane and Kutin \cite{kane_quantum_2010}, and independently by Meyer and Pommershein \cite{meyer_uselessness_2010}, that at least $\frac{d+1}{2}$ queries are needed to compute the coefficients with bounded error. Subsequently it was shown by Boneh and Zhandry \cite{hutchison_quantum-secure_2013} that $d$ quantum queries suffice, therefore giving a quantum-speedup of at least 1 query. Finally it was shown by Childs et. al. in 2016 \cite{childs_optimal_2016} that the lower bound $\frac{d+1}{2}$ is in fact tight: there exists a way to compute the coefficients with bounded error with only that many quantum queries. 

Here we present a diagrammatic account of Boneh and Zhandry's \cite{hutchison_quantum-secure_2013} algorithm to reduce the number of queries by $1$. Begin by noticing that a classical computer can reduce the degree-$d$ problem to a degree-$1$ problem in only $d-1$ classical queries. So it suffices to find a quantum algorithm that interpolates a degree-one polynomial in one single quantum query. 

So let $f(x) = ax + b$, with $a \in \mathbb F_q^\times$ and $b \in \mathbb F_q$. We begin by extending $f$ linearly to $F: \mathbb C^q \to \mathbb C^q$: 
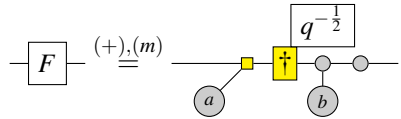
\begin{equation*}
    \begin{tikzpicture}
	\begin{pgfonlayer}{nodelayer}
		\node [style=gate] (0) at (0, 0) {$F$};
		\node [style=none] (1) at (-1, 0) {};
		\node [style=none] (2) at (1, 0) {};
	\end{pgfonlayer}
	\begin{pgfonlayer}{edgelayer}
		\draw (1.center) to (0);
		\draw (0) to (2.center);
	\end{pgfonlayer}
\end{tikzpicture}
 \overset{(\ref{eq:add}), (\ref{eq:mult})}{=} \input{./figures/F-constr.tikz}
\end{equation*}
and constructing the following quantum oracle: 
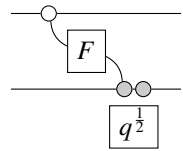
\begin{equation*}
    \input{./figures/oracle.tikz}
\end{equation*}
Inputting the uniform superposition to the first register, and the $0$ basis state to the second, we obtain: 
\begin{equation*}
    \psi_1 = \input{./figures/superposition-constr.tikz} = \input{./figures/superposition.tikz}
\end{equation*}

Now we define the following controlled unitary $U$: 
\begin{equation*}
    U = \,\, \input{./figures/big-U-0.tikz} + \sum_{y \neq 0} \input{./figures/big-U-1.tikz} 
\end{equation*}
note the H-spider bit of the nonzero terms is indeed unitary because $y \neq 0$. We then apply this unitary to the state $\psi_1$, intermediated by an H-gate on the bottom register: 
\begin{align*}
    \psi_2 \,\,\,\, = \input{./figures/psi2.tikz} = \,\,\,\,  \input{./figures/psi2-1-0.tikz} + \sum_{y \neq 0} \,\, \input{./figures/psi2-1.tikz}
\end{align*}

Now let's focus on the $y \neq 0$ terms in the above $\psi_2$, as that's where the bulk of the magic happens:
\begin{align*}
    \input{./figures/psi2-1.tikz} 
    \hspace{3mm} \overset{\ref{eq:basis-lollipop-def}}{=} \hspace{3mm} \input{./figures/psi2-2.tikz} \hspace{3mm}\overset{\ref{eq:fourier-copy-rule}}{=} \hspace{3mm}\input{./figures/psi2-3.tikz} \\[2mm] 
    \overset{\ref{eq:basis-lollipop-def}}{=} \hspace{3mm} \input{./figures/psi2-5.tikz} 
    \hspace{3mm} {=} \hspace{3mm} \input{./figures/psi2-6.tikz}
\end{align*}
So: 
\begin{equation*}
    \psi_2 = \input{./figures/psi2-1-0.tikz} +  \sum_{y \neq 0} \input{./figures/psi2-6.tikz}
\end{equation*}

The contents of the first term in $\psi_2$ is useless to us. So we now perform a non-demolition projective measurement on the second register. To do so, we introduce some classical registers. We now use bolded lines to denote quantum systems, viewed in the mixed-state formalism, and thin lines to denote classical systems viewed probabilistically. Let: 
\begin{equation*}
    \Psi_2 = double\left(\psi_2\right)
\end{equation*}
Where the "doubling" is the movement to the CPM construction, so that vectors become density matrices, and linear transformations become operators between density matrices, in Selinger's sense \cite{selinger_dagger_2007}. Define the projective measurement
\begin{equation*}
    P = \input{./figures/proj-0.tikz} + \sum_{y \neq 0} \input{./figures/proj-1.tikz}
\end{equation*}
Applying $P$ to the second register of $\Psi_2$ and measuring the outcome, we get 
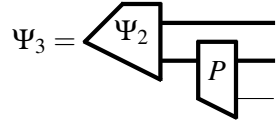
\begin{equation*}
    \Psi_3 = \input{./figures/psi3.tikz}
\end{equation*}
We post-select on the classical output for $1$. The probability that the projective measurement returns $0$ is just the probability that the linear function $f$ returns $0$ when the input is uniformly random. But since any linear function on $\mathbb F_q$ is bijective, so the probability that $f$ returns $0$ is just $1/q$. Thus the probability that $P$ returns $0$ is also $1/q$.  
\begin{equation*}
    \Psi_3 |_{y \neq 0} = \frac{q}{q-1} \input{./figures/Psi3-1.tikz} = \frac{q}{q-1} 
 \,\,double\left(\sum_{y \neq 0} \input{./figures/Psi2-6.tikz}\right)
\end{equation*}

Now we simply hit the second register with the inverse Fourier transform and measure, obtaining  
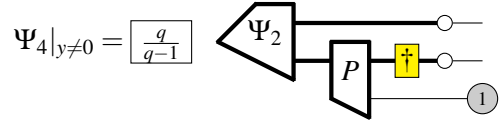
\begin{equation*}
    \Psi_4|_{y \neq 0} = \input{./figures/psi4.tikz}
\end{equation*}
The first register is sure to return $a$; whereas the probability that the second register returns $b$ is  
\begin{align*}
    \frac{q}{q-1} \,\,\, double \left(\sum_{y \neq 0} \input{./figures/prob-b-1.tikz}\right) = \frac{q}{q-1}\left(\sum_{y \neq 0} \frac{1}{q}\omega^{\tr(by) - \tr(by)}\right)^2 = \frac{q}{q-1}\left(\frac{q-1}{q}\right)^2 = \frac{q-1}{q}
\end{align*}
where the first equality uses corollary \ref{cor:lollipops_as_fourier_bases}. Thus the second register returns $b$ with a probability of $\frac{q-1}{q}$. 

So, to sum up, the quantum polynomial interpolation algorithm is: 
\begin{algo}
    Given a classical oracle $g: \mathbb F_q \to \mathbb F_q$ a polynomial of known degree $d$: 
    \begin{enumerate}
        \item Use $d-1$ classical queries to reduce the problem of finding the coefficients of $g$ to the problem of finding the coefficients of a degree-1 polynomial $f(x) = ax + b$ with unknown $a \in \mathbb F_q^\times$ and $b \in \mathbb F_q$. 
        \item Extend $f$ linearly to a transformation $F: \mathbb C^q \to \mathbb C^q$, and build the following quantum oracle: 
        \begin{center}
            \input{./figures/oracle-doubled.tikz}
        \end{center}
        \item Make a single query to the oracle with the uniform superposition as the first input and the zeroth basis state as the second input; and then apply the following operations: 
        \begin{center}
            \input{./figures/full-algo.tikz}
        \end{center}
        where 
        \begin{equation*}
            \hat U = \,\, double \left(\,\,\,\,\input{./figures/big-U-0.tikz} + \sum_{y \neq 0} \,\,\input{./figures/big-U-1.tikz}\,\,\right)
        \end{equation*} 
        \begin{equation*}
            P = \input{./figures/proj-0.tikz} + \sum_{y \neq 0} \input{./figures/proj-1.tikz}
        \end{equation*}
        \item The right-most classical output is $0$ with probability $1/q$. In that case, abort and restart. Else, the first classical output is $a$ with probability $1$, and the second classical output is $b$ with probability $(q-1)/q$. 
    \end{enumerate}
\end{algo}

\section{Conclusion}
\label{section:conclusion}

We have presented an extension of the ZH calculus to Hilbert spaces of prime-power dimensionality $q = p^t$, such that the generalized AND gate constructed from the H-boxes corresponds to field multiplication in the finite field $\mathbb F_q$. We have shown that this finite field ZH calculus, with phases, is universal over matrices with entries in the ring $\mathbb Z[\omega]$. The usefulness of an extension of this form is demonstrated by its ability to reason elegantly about the quantum polynomial interpolation algorithm. 

An obvious direction for future research is to consider the completeness of the rewrite rules presented in section \ref{subsection:rewrite-rules}. A rule that was present in \cite{backens_completeness_2023}, corresponding to the fact that a field has no zero-divisor, is missing in our ruleset. It remains to be seen whether there is an adequate replacement rule to be found for the finite field ZH, or whether the rule is in fact redundant. 

Another direction is to consider whether the phase-\textit{free} finite field ZH calculus is also universal. To prove universality for the phase-free calculus, one would need what is known as a "successor gadget" helping one walk over the necessary phases for a single-output H-lollipop \cite{roy_qudit_2023}. The technique for constructing such a gadget, and indeed the $0$-phased H-lollipop itself, requires some further work. 

\bibliographystyle{eptcs}
\bibliography{references}

\appendix
\section{Soundness of Finite Field Arithmetics in ZH}
\label{append:arithmetics}

\begin{claim}
    (Proposition \ref{prop:neg}) The 1-to-1 X-spider is the same as applying the H-box twice, and acts on the computational basis as negation: 
    \begin{equation}
        \begin{tikzpicture}
	\begin{pgfonlayer}{nodelayer}
		\node [style=X dot] (0) at (0, 0) {};
		\node [style=none] (1) at (-1, 0) {};
		\node [style=none] (2) at (1, 0) {};
	\end{pgfonlayer}
	\begin{pgfonlayer}{edgelayer}
		\draw (0) to (1.center);
		\draw (0) to (2.center);
	\end{pgfonlayer}
\end{tikzpicture}
 = \begin{tikzpicture}
	\begin{pgfonlayer}{nodelayer}
		\node [style=hadamard] (0) at (-0.5, 0) {};
		\node [style=hadamard] (1) at (0.5, 0) {};
		\node [style=none] (3) at (1.5, 0) {};
		\node [style=none] (4) at (-1.5, 0) {};
	\end{pgfonlayer}
	\begin{pgfonlayer}{edgelayer}
		\draw (0) to (1);
		\draw (1) to (3.center);
		\draw (4.center) to (0);
	\end{pgfonlayer}
\end{tikzpicture}
 = \sum_{i \in \mathbb F_q} |-i\rangle \langle i|  \tag{$\neg$}
    \end{equation}
\end{claim}

\begin{proof}
    \begin{align*}
        \begin{tikzpicture}
	\begin{pgfonlayer}{nodelayer}
		\node [style=X dot] (0) at (0, 0) {};
		\node [style=none] (1) at (-1, 0) {};
		\node [style=none] (2) at (1, 0) {};
	\end{pgfonlayer}
	\begin{pgfonlayer}{edgelayer}
		\draw (0) to (1.center);
		\draw (0) to (2.center);
	\end{pgfonlayer}
\end{tikzpicture}
 &= \input{./figures/x-neg-constr.tikz} \\ 
        &\overset{(id)}{=} \begin{tikzpicture}
	\begin{pgfonlayer}{nodelayer}
		\node [style=hadamard] (0) at (-0.5, 0) {};
		\node [style=hadamard] (1) at (0.5, 0) {};
		\node [style=none] (3) at (1.5, 0) {};
		\node [style=none] (4) at (-1.5, 0) {};
	\end{pgfonlayer}
	\begin{pgfonlayer}{edgelayer}
		\draw (0) to (1);
		\draw (1) to (3.center);
		\draw (4.center) to (0);
	\end{pgfonlayer}
\end{tikzpicture}
 \\ 
        &= \frac{1}{q}\sum_{i, j, k \in \mathbb F_q} \omega^{\tr(ij)} \omega^{\tr(jk)} |k\rangle \langle i | \\ 
        &= \frac{1}{q}\sum_{i, j, k \in \mathbb F_q} \omega^{\tr(j(i+k))} |k\rangle \langle i | \\ 
        &\overset{\ref{lemma:sum-on-circle}}{=} \sum_{i \in \mathbb F_q}|-i\rangle \langle i |
    \end{align*}
\end{proof}

\begin{claim}
    (Proposition \ref{prop:add}) The 2-to-1 X-spider implements addition in $\mathbb F_q$, and the 0-to-1 X-spider is the additive identity $0 \in \mathbb F_q$: 
    \begin{equation}
         \input{./figures/Addition.tikz} = |i+j\rangle \tag{$+$}
    \end{equation}
    \begin{equation}
        \begin{tikzpicture}
	\begin{pgfonlayer}{nodelayer}
		\node [style=X dot] (0) at (0, 0) {};
		\node [style=none] (1) at (1, 0) {};\
        \node [style=scalar] (2) at (-1.5, 0) {$q^{-\frac{1}{2}}$}; 
	\end{pgfonlayer}
	\begin{pgfonlayer}{edgelayer}
		\draw (0) to (1.center);
	\end{pgfonlayer}
\end{tikzpicture}
 = |0\rangle \tag{$0$}
    \end{equation}
\end{claim}

\begin{proof}
    \begin{align*}
        \input{./figures/Addition.tikz} &= \input{./figures/Addition-constr.tikz} \\ 
        &= \frac{1}{q}\sum_{k, l \in \mathbb F_q} \omega^{-\tr(lk)} \omega^{\tr(ki)} \omega^{\tr(kj)} |l\rangle \\ 
        &= \frac{1}{q}\sum_{k, l\in \mathbb F_q} \omega^{\tr(k(i + j - l))} |l\rangle \\ 
        &\overset{\ref{lemma:sum-on-circle}}{=} |i+j\rangle \\
        \begin{tikzpicture}
	\begin{pgfonlayer}{nodelayer}
		\node [style=X dot] (0) at (0, 0) {};
		\node [style=none] (1) at (1, 0) {};\
        \node [style=scalar] (2) at (-1.5, 0) {$q^{-\frac{1}{2}}$}; 
	\end{pgfonlayer}
	\begin{pgfonlayer}{edgelayer}
		\draw (0) to (1.center);
	\end{pgfonlayer}
\end{tikzpicture}
 &= \begin{tikzpicture}
	\begin{pgfonlayer}{nodelayer}
		\node [style=Z dot] (0) at (-1, 0) {};
		\node [style=hadamard] (1) at (0, 0) {};
		\node [style=none] (2) at (1, 0) {};
		\node [style=scalar] (3) at (-2.5, 0) {$q^{-\frac{1}{2}}$};
	\end{pgfonlayer}
	\begin{pgfonlayer}{edgelayer}
		\draw (0) to (1);
		\draw (1) to (2.center);
	\end{pgfonlayer}
\end{tikzpicture}
 \\ 
        &= \frac{1}{q} \sum_{k, l} \omega^{\tr(lk)} |l\rangle \\ 
        &\overset{\ref{lemma:sum-on-circle}}{=} |0\rangle. 
    \end{align*}
\end{proof}

\begin{claim}
    (Proposition \ref{prop:mult}) The 2-to-1 H-spider followed by the inverse Fourier transform, implements multiplication in $\mathbb F_q$. and the 0-to-1 H-spider followed by the inverse Fourier transform is the multiplicative identity $1 \in \mathbb F_q$: 
    \begin{equation}
        \input{./figures/Multiplication.tikz} = |i\cdot j\rangle \tag{$\cdot$}
    \end{equation}
    \begin{equation}
        \begin{tikzpicture}
	\begin{pgfonlayer}{nodelayer}
		\node [style=hadamard] (0) at (-0.5, 0) {};
		\node [style=hadamard] (1) at (0.5, 0) {$\dag$};
		\node [style=none] (2) at (1.5, 0) {};
	\end{pgfonlayer}
	\begin{pgfonlayer}{edgelayer}
		\draw (0) to (1);
		\draw (1) to (2.center);
	\end{pgfonlayer}
\end{tikzpicture}
 = |1\rangle \tag{$1$}
    \end{equation}
\end{claim}

\begin{proof}
    \begin{align*}
        \input{./figures/Multiplication.tikz} &= \frac{1}{q}\sum_{k, l \in \mathbb F_q} \omega^{-\tr(l k)} \omega^{\tr(k ij)} |l\rangle \\ 
        &= \frac{1}{q} \sum_{k, l \in \mathbb F_q} \omega^{\tr((l - ij) k)} |l\rangle \\ 
        &\overset{\ref{lemma:sum-on-circle}}{=} \frac{1}{q} \cdot q |i \cdot j\rangle \\ 
        &= |i\cdot j\rangle \\ 
        \begin{tikzpicture}
	\begin{pgfonlayer}{nodelayer}
		\node [style=hadamard] (0) at (-0.5, 0) {};
		\node [style=hadamard] (1) at (0.5, 0) {$\dag$};
		\node [style=none] (2) at (1.5, 0) {};
	\end{pgfonlayer}
	\begin{pgfonlayer}{edgelayer}
		\draw (0) to (1);
		\draw (1) to (2.center);
	\end{pgfonlayer}
\end{tikzpicture}
 &= \frac{1}{q} \sum_{k, l \in \mathbb F_q} \omega^{-\tr(l k)} \omega^{\tr(k1)} |l\rangle \\ 
        &= \frac{1}{q} \sum_{k, l \in \mathbb F_q} \omega^{\tr(k (1 - l))} |l\rangle \\ 
        &\overset{\ref{lemma:sum-on-circle}}{=} |1\rangle. 
    \end{align*}
\end{proof}





\pagebreak
\section{Soundness of Rewrite Rules} 
\label{append:rewrite}

\begin{claim}
    The following equations hold when the ZH calculus is interpreted as linear maps: 
    \begin{figure*}[h]
    \centering
    \begin{subfigure}[b]{0.4\textwidth}
        \centering
        \input{./figures/rule-zs-lhs.tikz} = \input{./figures/rule-zs-rhs.tikz} 
        \caption*{(zs)}
    \end{subfigure}
    \begin{subfigure}[b]{0.4\textwidth}
        \centering
        \input{./figures/rule-hs-lhs.tikz} = \input{./figures/rule-hs-rhs.tikz} 
        \caption*{(hs)}
    \end{subfigure}
    \begin{subfigure}[b]{0.4\textwidth}
        \centering
        \input{./figures/rule-id-lhs.tikz} = \input{./figures/Wire.tikz} 
        \caption*{(id)}
    \end{subfigure}
    \begin{subfigure}[b]{0.4\textwidth}
        \centering        \input{./figures/rule-ch-lhs.tikz} = \input{./figures/rule-ch-rhs.tikz} 
        \caption*{(ch)}
    \end{subfigure}
    \begin{subfigure}[b]{0.4\textwidth}
        \centering 
        \input{./figures/rule-spl-lhs.tikz} = \input{./figures/Wire.tikz} 
        \caption*{(spl)}
    \end{subfigure}
        \begin{subfigure}[b]{0.4\textwidth}
        \centering
        \input{./figures/copy-thru-white-lhs.tikz} = \input{./figures/copy-thru-white-rhs.tikz} 
        \caption*{(cp)}
    \end{subfigure}
    \begin{subfigure}[b]{0.4\textwidth}
        \centering
        \input{./figures/rule-ba1-lhs.tikz} = \input{./figures/rule-ba1-rhs.tikz} 
        \caption*{(ba1)}
    \end{subfigure}
    \begin{subfigure}[b]{0.4\textwidth}
        \centering
        \input{./figures/rule-ba2-lhs.tikz} = \input{./figures/rule-ba2-rhs.tikz} 
        \caption*{(ba2)}
    \end{subfigure}
    \begin{subfigure}[b]{0.4\textwidth}
        \centering
        \input{./figures/rule-pm-lhs.tikz} = \input{./figures/rule-pm-rhs.tikz}
        \caption*{(pm)}
    \end{subfigure}
\end{figure*}
\end{claim}
\begin{proof}
    (Of the zs rule) The (zs) rule encodes the information that copying is coassociative: 
    \begin{align*}
        \input{./figures/rule-zs-lhs.tikz} &= \left(\sum_{j \in \mathbb F_q} |j\rangle ^{\otimes m} \langle j | \right) \circ \left(\sum_{i \in \mathbb F_q} |i \rangle \langle i |^{\otimes n}\right) \\ 
        &= \sum_{j \in \mathbb F_q} |j \rangle ^{\otimes m} \langle j | j \rangle \langle j |^{\otimes n} \\ 
        &= \sum_{j \in \mathbb F_q} |j\rangle ^{\otimes m} \langle j |^{\otimes n} \\ 
        &= \input{./figures/rule-zs-rhs.tikz}. 
    \end{align*}
\end{proof}

\begin{proof}
    (Of the hs rule) The (hs) rule encodes the information that multiplication is associative: 
    \begin{align*}
        \input{./figures/rule-hs-lhs.tikz} &= q^{-\frac{3}{2}}\sum_{\substack{j_1, ..., j_m \\ i_1, ..., i_n \\ k, l} \in \mathbb F_q} \omega^{\tr(i_1...i_n k ) - \tr(kl) + \tr(l j_1...j_m)} |j_1\rangle ... |j_m\rangle \langle i_1 | ... \langle i_n | \\ 
        &= q^{-\frac{3}{2}}\left(\sum_{j_1, ..., j_m, l} \omega^{\tr(lj_1...j_m)} |j_1\rangle ... |j_m \rangle \langle l |\right) \circ \left(\sum_{i_1, ..., i_m, k, l} \omega^{\tr((i_1...i_m - l )k)} |l\rangle \langle i_1 | ... \langle i_n |\right) \\ 
        &= q^{-\frac{1}{2}}\left(\sum_{j_1, ..., j_m, l} \omega^{\tr(lj_1...j_m)} |j_1\rangle ... |j_m \rangle \langle l |\right) \circ  \left(\sum_{i_1, ..., i_m} |i_1 \cdot ... \cdot i_n \rangle \langle i_1 | ... \langle i_n |\right) \\ 
        &= q^{-\frac{1}{2}} \sum_{\substack{i_1, ..., i_n \\ j_1, ..., j_m}} \omega^{\tr(i_1 ... i_n j_1 ... j_m)} |j_1 \rangle ... |j_m \rangle \langle i_1 | ... \langle i_n | \\ 
        &= \input{./figures/rule-hs-rhs.tikz}. 
    \end{align*}
\end{proof}

\begin{proof}
    (Of the id rule) The (id) rule encodes the resolution of identity: 
    $$\begin{tikzpicture}
	\begin{pgfonlayer}{nodelayer}
		\node [style=Z dot] (0) at (0, 0) {};
		\node [style=none] (1) at (-1, 0) {};
		\node [style=none] (2) at (1, 0) {};
	\end{pgfonlayer}
	\begin{pgfonlayer}{edgelayer}
		\draw (1.center) to (0);
		\draw (0) to (2.center);
	\end{pgfonlayer}
\end{tikzpicture}
 = \sum_{j \in \mathbb F_q} |j\rangle \langle j | = \mathbf 1$$
\end{proof}

\begin{proof}
    (Of the ch rule) The (ch) rule encodes the information that the character of $\mathbb F_q$ is $p$: 
    \begin{align*}
        \input{./figures/rule-ch-lhs.tikz} &= q^{-\frac{1}{2}}\sum_{j \in \mathbb F_q} |p \cdot j \rangle \langle j | \\ 
        &= q^{-\frac{1}{2}}\sum_{j \in \mathbb F_q} |0\rangle \langle j | \\ 
        &= q^{-1}\left(q^{\frac{1}{2}} |0\rangle \right) \circ \left(\sum_{j \in \mathbb F_q} \langle j |\right) \\
        &= \input{./figures/rule-ch-rhs.tikz}
    \end{align*}
\end{proof}

\begin{proof}
    (Of the spl rule) The (spl) rule encodes the information that every element of $\mathbb F_q$ is a root of the polynomial $x^q - x$: 
    \begin{align*}
        \input{./figures/rule-spl-lhs.tikz} &\overset{\ref{prop:mult}}{=} \sum_{j \in \mathbb F_q} |j^q \rangle \langle j | \\ 
        &= \sum_{j \in \mathbb F_q} |j \rangle \langle j | \\ 
        &= \mathbf 1. 
    \end{align*}
    where the second equality follows from the fact that every element of $\mathbb F_q$ is a root of the polynomial $x^q - x$. 
\end{proof}

\begin{proof}
    (Of the cp rule) The (cp) rule encodes the information that the Z-spider copies computational basis states: 
    \begin{align*}
        \input{./figures/copy-thru-white-lhs.tikz} &= \left(\sum_{i \in \mathbb F_q} |i \rangle |i \rangle \langle i |\right) \circ \left(q^{\frac{1}{2}} |j \rangle\right) \\ 
        &= q^{\frac{1}{2}} |j \rangle | j \rangle \\ 
        &= \input{./figures/copy-thru-white-rhs.tikz}. 
    \end{align*}
\end{proof}

\begin{proof}
    (Of the ba1 rule) The (ba1) rule encodes the fact that addition is deterministic, and so commutes with copying. Begin by proving the rule for the special case in which $n = 1$: 
    \begin{align*}
        \input{./figures/neg-copy.tikz} &\overset{\ref{prop:neg}}{=} \sum_{j \in \mathbb F_q} |-j\rangle^{\otimes m} \langle j | \\
        &= \input{./figures/copy-neg.tikz}. 
    \end{align*}
    Then, from the fact that the x-spiders implement addition (\ref{prop:add}), and the fact that addition is deterministic, we have 
    $$\input{./figures/add-copy.tikz} = \input{./figures/copy-add.tikz}$$
    Post-composing on both sides with x-spiders we have 
    $$\input{./figures/add-copy-neg.tikz}  = \input{./figures/copy-add-neg.tikz}$$
    Applying the special case rule for $n = 1$ to the left-hand-side, and simplifying, we get 
    $$\input{./figures/rule-ba1-lhs.tikz} = \input{./figures/rule-ba1-rhs.tikz}$$
\end{proof}

\begin{proof}
    (Of the ba2 rule) The (ba2) rule encodes the fact that multiplication is deterministic, and so commutes with copying. Since the inverse-Fourier transform is unitary, we just have to prove the equality after post-composing both sides with inverse-Fourier transforms. The left-hand-side becomes 
    \begin{align*}
        \input{./figures/ba2-1.tikz} &= \input{./figures/ba2-2.tikz} \\ 
        &= \input{./figures/ba2-3.tikz} \\ 
        &= \input{./figures/ba2-4.tikz} \\ 
        &= \left(\sum_{k} |-k\rangle ^{\otimes m} \langle k |\right) \circ \left(|-i_1\cdot ...\cdot i_n\rangle \langle i_1 | ... \langle i_n |\right) \\ 
        &= \sum_{i_1, ..., i_n} |i_1 \cdot ... \cdot i_n \rangle ^{\otimes m} \langle i_1 | ... \langle i_n | 
    \end{align*}
    While the right-hand-side becomes 
    \begin{align*}
        \input{./figures/ba2-5.tikz} &= \sum_{i_1, ..., i_n} |i_1 \cdot ... \cdot i_n \rangle ^{\otimes m} \langle i_1 | ... \langle i_n |
    \end{align*}
\end{proof}

\begin{proof}
    (Of the pm rule) The (pm) rule is a special case of the fact that the Z-spider implements Schur product of states, and thereby serves as a product on the ring of phases $R$: 
    \begin{align*}
        \input{./figures/rule-pm-lhs.tikz} &= \sum_{j \in \mathbb F_q} r_1^{\tr(j)} r_2^{\tr(j)} |j \rangle \\ 
        &= \sum_{j \in \mathbb F_q} (r_1 r_2)^{\tr(j)} |j\rangle \\ 
        &= \begin{tikzpicture}
	\begin{pgfonlayer}{nodelayer}
		\node [style=hadamard] (0) at (0, 0) {$r_1 r_2$};
		\node [style=none] (1) at (1, 0) {};
	\end{pgfonlayer}
	\begin{pgfonlayer}{edgelayer}
		\draw (0) to (1.center);
	\end{pgfonlayer}
\end{tikzpicture}

    \end{align*}
\end{proof}

\end{document}